\documentclass{amsart}

\usepackage{tikz-cd}

\usepackage{verbatim}
\usepackage{enumerate}
\usepackage{hyperref}
\usepackage{mathtools}
\usepackage{epigraph}
\usepackage{amssymb}

\newcommand{\ndash}{\nobreakdash-\hspace{0pt}}
\newcommand{\Ndash}{\nobreakdash--}

\newcommand{\ii}{{\mathrm{i}}}
\newcommand{\dd}{{\mathrm{d}}}

\newcommand{\EE}{\mathrm{e}}

\newcommand{\Div}{{\mathrm{div}}}

\newcommand{\frX}{\mathfrak{X}}

\newtheorem{Thm}{Theorem}[section]

\newtheorem*{Thm*}{Theorem}
\newtheorem*{Lem*}{Lemma}

\theoremstyle{remark}
\newtheorem*{Ack}{Acknowledgment}

\theoremstyle{definition}
\newtheorem{Rem}[Thm]{Remark}
\newtheorem{Dig}[Thm]{Digression}
\newtheorem{Def}[Thm]{Definition}

\newtheorem{Not}[Thm]{Notation}

\newcommand{\bbR}{{\mathbb{R}}}

\newcommand{\de}{\partial}

\newcommand{\bD}{{\boldsymbol{D}}}

\newcommand{\calV}{\mathcal{V}}

\newcommand{\frg}{{\mathfrak{g}}}

\def\gpd{\,\lower1pt\hbox{$\longrightarrow$}\hskip-.24in\raise2pt
               \hbox{$\longrightarrow$}\,}

\let\Tilde=\widetilde

\let\Hat=\widehat

\newcommand\qq{}
\newcommand\cmp[1]{{\qq Commun.\ Math.\ Phys.\ \bf #1}}

\newcommand\pl[1]{{\qq Phys.\ Lett.\ \bf #1}}
\newcommand\np[1]{{\qq Nucl.\ Phys.\ \bf #1}}

\newcommand\phr[1]{{\qq Phys.\ Rev.\ \bf #1}}

\newcommand\ijmp[1]{{\qq Int.\ J. Mod.\ Phys.\ \bf #1}}

\newcommand\adm[1]{{\qq Adv.\ Math.\ \bf #1}}

\newcommand\atmp[1]{{\qq Adv.\ Theor.\ Math.\ Phys.\ \bf #1}} 
\newcommand\ahp[1]{{\qq Ann.\ Henri Poincar\'e \bf #1}}

\begin{document}
\title{Phase Space for Gravity with Boundaries}

\begin{abstract}
This explanatory note, based on the geometrical method by Kijovski and Tulczyjew, describes the construction of the reduced phase space of Lagrangian field theories, i.e., the correct space of initial conditions with its symplectic structure. Several examples and, in particular, the case of four-dimensional gravity in the coframe formalism (Palatini--Cartan theory) are analyzed.
\end{abstract}

\author{Alberto S. Cattaneo}
\address{Institut f\"ur Mathematik, Universit\"at Z\"urich\\
Winterthurerstrasse 190, CH-8057 Z\"urich, Switzerland}  
\email{cattaneo@math.uzh.ch}


\thanks{I acknowledge partial support of the SNF Grant No.\ 200020\_192080 and of the Simons Collaboration on Global Categorical Symmetries.
This research was (partly) supported by the NCCR SwissMAP, funded by the Swiss National Science Foundation. This article is based upon work from COST Action 21109 CaLISTA, supported by COST (European Cooperation in
Science and Technology) \href{www.cost.eu}{(www.cost.eu)}. This article is commissioned by the Encyclopedia of Mathematical Physics, edited by M. Bojowald and R. J. Szabo, to be published by Elsevier \cite{ency}.}





\maketitle

Key points 
\begin{itemize}
\item Definition of Lagrangian field theories
\item Euler--Lagrange equations and Noether $1$-form
\item Boundary structure
\item Evolution relations
\item Cauchy data and the reduced phase space
\item Some examples: scalar field, electromagnetism, Palatini--Cartan theory.
\end{itemize}

\tableofcontents


\section{Introduction}
One useful description of a mechanical system is in terms of hamiltonian evolution. This is given by a symplectic manifold and a hamiltonian function---the evolution being given by the flow of its hamiltonian vector field. This description is particulary interesting because it is, in principle, ready-made for quantization, with the symplectic manifold being turned to a Hilbert space and the hamiltonian function to a hamiltonian operator (how this is done, if anyway possible, is another story).

Very often a fundamental system is however introduced via an action functional and its associated Euler--Lagrange (EL) equations. Recasting it as a hamiltonian system---in particular in the case of gravity---is the goal of this note.

The (reduced) phase space of the system refers to the symplectic manifold on which one formulates
the hamiltonian evolution of the system, whenever possible, once one chooses a time axis. Roughly speaking, it is the space of Cauchy data for the EL equations. The main issue is that the latter may have also a nonevolutionary content, which then gives rise to constraints. This is typically the case of gauge theories and of gravity.

The traditional way of proceding is due to Dirac \cite{Dirac50}. This is an algebraic method with the goal of extending the Legendre transform to the case of degenerate Lagrangians. The method is iterative, with the construction of primary and a sequence of secondary constraints, and can in principle break down at every step. 

In this note we will focus on a more geometric approach due to Kijovski and Tulczyjew (KT) \cite{KT} (we also rely on the presentation in \cite{CMR11}).\footnote{A related construction, which we will not discuss here, goes under the name of ``covariant phase space'' (see \cite{WZ99,BB01,FOPS21,MBV21,MBV22}, and also
 \cite{nLabps} for an overview).  In this approach, the space of solutions, equipped with its symplectic structure, is usually constructed via the bivariational complex (see, e.g., \cite{BHL09} and references therein).}
The central idea is to reformulate the space of solutions to the EL equations as an isotropic relation between the spaces of initial and final field configurations 
naturally endowed with a closed $2$\ndash form. The flow property is here replaced by the set\ndash theoretic composition of relations, which only in particular cases are graphs of a flow. The initial (or final) space is then naturally endowed with a subset of Cauchy data---the points that can belong to the relation for an arbitrary small interval of time. As the restriction\footnote{\label{f-restriction}Following the common terminology in symplectic geometry, by ``restriction'' we mean the ``pullback by the inclusion map.''} of the $2$\ndash form to the space of Cauchy data is usually degenerate, one has to consider an appropriate quotient (by the null directions of the $2$\ndash form) 
to recover the reduced phase space as a symplectic manifold. It may of course happen that the quotient is singular. However, there is a natural way of doing this reduction in stages. 
In most physical theories, the only problematic quotient turns out to be the last one: the symplectic reduction of a coisotropic (i.e., defined by first\ndash class constraints) submanifold of a symplectic manifold of (appropriately defined) boundary fields. 
To deal with the last reduction, there are appropriate techniques (e.g., the Batalin--Fradkin--Vilkovisky formalism (BFV) \cite{FV75,FK77}, 
which will not be discussed in this note). 

One advantage of the KT method is that it is rather flexible. One does not really have to consider evolution between an initial and a final time, but one can focus on the initial time only. One can even drop the condition that the initial submanifold should be space-like---the formalism working in more general situations---and one does not have to consider cylinders (a product of a manifold by an interval), but one can consider more general manifolds with boundary. One further advantage is 
that this method can be readily applied to the extension of the BV formalism \cite{BV77,BV81} to manifolds with boundary (and possibly lower\ndash dimensional strata), producing the BV-BFV formalism discussed in \cite{CMR14}.

We will start with a warm up dealing with simples cases. Next we will formulate the general theory and will apply it to the already presented examples. Finally, we will discuss gravity 
in the coframe (a.k.a.\ Palatini--Cartan) formulation.

\begin{Ack}
I thank J. Huerta, M. Schiavina, S. Speziale, and M. Tecchiolli
for useful comments on a first draft.
\end{Ack}

\section{Some preliminary examples}\label{s:prel-exa}
We start reviewing some simple examples that will motivate the general theory. We will review in Section~\ref{s-exaafterKT} how they are described within the general KT method.

\subsection{Mechanics}\label{s:mechanics}
Let us start with the familiar example of mechanics in one dimension. The action functional is
\[
S[q]=\int_a^b \left(\frac12m\dot q^2-V(q)\right)dt,
\]
where $q$ (the ``field'') is a path $[a,b]\to\bbR$, $m$ is the mass, and $V$ is the potential energy. We want to compute a variation of $S$ isolating in the integral variations of $q$ but not of its derivative(s). We can do this integrating by parts:
\[
\delta S = -\int_a^b \left(m\ddot q+V'(q)\right)\delta q\,dt
+m\dot q\,\delta q\Big |_a^b.
\]
The bulk term---i.e., the integral over the interval---contains the EL equations as the coefficient of $\delta q$:
\[
m\ddot q+V'(q)=0.
\]
The boundary term can be reinterpreted as the difference, at times $b$ and $a$,
of a $1$\ndash form, on the space of positions $q$ and velocities $v=\dot q$, that we denote by $\alpha$ and call the \textsf{Noether $1$\ndash form}:\footnote{More precisely, we regard this $1$\ndash form as defined on the tangent bundle  $T\bbR=\bbR\times\bbR\ni(v,q)$ of the space of fields at $a$ (or $b$), 
so $\delta$ becomes the usual de~Rham differential. 
The boundary term in the variation of the action is more explicitly written as
\[
m\dot q\,\delta q\Big |_a^b = \pi_b^*\alpha - \pi_a^*\alpha,
\]
where $\pi_a$ and $\pi_b$ are the maps from the space of paths to $\bbR\times\bbR$ that send the path $q$ to $(\dot q(a),q(a))$ and to $(\dot q(b),q(b))$, respectively.}
\[
\alpha = mv\,\delta q.
\]
We then take its differential
\[
\omega \coloneqq\delta\alpha = m\,\delta v\,\delta q,
\]
which is an example of a \textsf{symplectic form}. This simply means that $\omega$ is closed (i.e., $\delta\omega=0$) and nondegenerate (i.e., hamiltonian vector fields are uniquely determined). 
We call the space $T\bbR$ of positions and velocities together with $\omega$ the \textsf{phase space} of the theory.\footnote{Note that the theory produces two copies of the phase space, one for the initial time $a$ and the other for the final time $b$.} 
(Usually, one calls phase space the space $T^*\bbR$ of positions $q$ and momenta $p=mv$ with $\omega=\delta p\,\delta q$, but 
this is just a change of variables.) As the EL equations are of second order, the phase space $T\bbR$ is the space of possible initial conditions, a.k.a.\ Cauchy data. Assuming the existence of global solutions, this may also be regarded as the space of solutions to the EL equations.

The dynamics is then given by the Hamiltonian $H=\frac12mv^2+V(q)$ on $T\bbR$ (or, as usual, $H=\frac{p^2}{2m}+V(q)$ on $T^*\bbR$) in the sense
that the time evolution is governed by the first-order ODE $(\dot v,\dot q)=X$ where $X=(-\frac{V'}m,v)$ is the unique vector field satisfying
\[
\iota_X\omega+\dd H = 0
\]
with $\iota$ denoting contraction. The vector field $X$ is called the hamiltonian vector field of $H$. 
The graph $L_{[a,b]}$ of the flow of $X$ can be shown to be a lagrangian submanifold of $T\bbR\times T\bbR$, upon changing the sign of the symplectic form on one of the two factors. This is an example of what we will in general call the evolution relation of the system.

\subsection{A degenerate example}\label{s:degenerate}
Consider now, as action functional, the euclidean length of a curve:
\[
S[q] = \int_a^b \|\dot q\|\;dt = \int_a^b\sqrt{\dot q\cdot\dot q}\;dt,
\]
where $q\colon[a,b]\to\bbR^n$ is a path. The first, minor, remark concerning this example is that, to avoid singularities, we have to restrict the space of paths to the regular ones, i.e., those with nonvanishing speed everywhere (this is an open condition, not a constraint). If we compute the variation of the action, we get
\[
\delta S = -\int_a^b\frac\dd{\dd t}\frac{\dot q}{\|\dot q\|}\,\delta q\; dt+
\frac{\dot q\cdot\delta q}{\|\dot q\|}\Bigg |_a^b.
\]
The EL equations are then
\[
\frac\dd{\dd t}\frac{\dot q}{\|\dot q\|}=0,
\]
which are solved by straight regular paths, with an unspecified parametrization. The Noether $1$\ndash form is
\[
\alpha = u\cdot \delta q,
\]
where we have introduced the notation $u$ for the normalized velocity $\frac{\dot q}{\|\dot q\|}$. The corresponding $2$\ndash form
\[
\omega = \delta u \cdot \delta q
\]
is degenerate. To look for its kernel, consider a vector field 
\[
X=X_q\cdot\frac\de{\de q}+X_u\cdot\frac\de{\de u}
\]
and compute the contraction
\[
\iota_X\omega = X_u\cdot\delta q - X_q\cdot\delta u.
\]
The first term may vanish only for $X_u=0$ because the components $\delta q^i$ of $\delta q$ are independent. On the other hand, since $u$ is normalized, we have $u\cdot u = 1$ which, differentiated, yields the relation $u\cdot\delta u=0$.
Since $\delta u$ is orthogonal to $u$, the last term in $\iota_X\omega$ vanishes exactly when $X_q$ is parallel to $u$.

To formulate evolution in hamiltonian terms, we have to mod out the kernel of the $2$\ndash form $\omega$. 
This reduction then consists in modding $q$ out along the $u$ direction. 
The reduced space is the  union over 
$u\in S^{n-1}$, the $(n-1)$\ndash dimensional sphere, of the spaces
\[
F_u = \{q\in\bbR^n\}/[q\sim q+\lambda u,\ \lambda\in\bbR].
\]
As a representative for each equivalence class, one can choose a vector orthogonal to $u$. Therefore, we may identify
$F_u$ with 
\[
T_u \coloneqq \{q\in\bbR^n:q\cdot u = 0\}.
\]
But this is just the tangent space of the sphere (viewed as submanifold of $\bbR^n$) at the point $u$.
We then have that the reduced phase space is $\cup_{u\in S^{n-1}}T_u=TS^{n-1}$, the tangent bundle of the sphere. The symplectic structure is $[\delta u]\cdot \delta q$. 
Using the round metric, one can actually identify $TS^{n-1}$ with $T^*S^{n-1}$. Under this transformation, the symplectic form is mapped to the canonical one, $\delta p\cdot\delta q$.

Finally note that the symplectic manifold $TS^{n-1}$ constructed above can also be regarded as the space of solutions to the EL equations, i.e., the space of oriented straight lines in $\bbR^n$: an oriented line is in fact specified by its direction (the unit tangent vector $u\in S^{n-1}$) and by its position in space
(a point $q\in\bbR^n$ modulo translations along $u$). 
Also note that these data 
do not change along the line, so the evolution on the reduced phase space $TS^{n-1}$ is given by the identity map, and the Hamiltonian vanishes.

\subsection{The scalar field}\label{s:scalar}
Let us now consider one example from field theory in $d$ dimensions. To a $d$\ndash dimensional Lorentzian manifold $(M,g)$ we associate the space of fields $C^\infty(M)$ and the action functional
\[
S_M[\phi]=\frac12\int_M g^{\mu\nu}\,\de_\mu\phi\,\de_\nu\phi\,\sqrt{|\det g|}\;\dd^dx,
\]
where $(g^{\mu\nu})$ are the components of the inverse of $g$.
The EL equation is the wave equation on $(M,g)$:
 \[
 \de_\mu(g^{\mu\nu}\de_\nu\phi\sqrt{|\det g|})=0.
 \]

We consider $M=\Sigma\times I$, where
$\Sigma$ is space-like and $I=[a,b]$ is a time interval. We also assume for simplicity that near the boundary the metric splits as
\[
g=-(\dd x^0)^2+h_{ij}\dd x^i\dd x^j,
\]
where $0$ denotes the time direction and $(h_{ij})$ is a euclidean 
metric. 
When we compute the variation of $S_M$, we get a boundary term---the Noether $1$\ndash form---on each of the two boundary copies of $\Sigma$:
\[
\alpha^\de_\Sigma = \int_\Sigma \phi_0\,\delta\phi\,\sqrt{\det h}\;\dd^{d-1}x,
\]
where $\phi_0$ denotes 
$\de_0\phi$.

In the symplectic setting, we  consider $\phi$ and $\phi_0$ as two independent functions on $\Sigma$: the initial configuration and the initial velocity. 
The space $F^\de_\Sigma$, with coordinates $\phi$ and $\phi_0$, can be identified with the space of solutions to the wave equation on
$\Sigma\times [a,b]$ with $\phi$ and $\phi_0$ as initial conditions at time $a$.
The space  $F^\de_\Sigma$ with symplectic form
\begin{equation}\label{e:symplscalar}
\omega^\de_\Sigma\coloneqq
\delta\alpha^\de_\Sigma = \int_\Sigma \delta\phi_0\,\delta\phi\,\sqrt{\det h}\;\dd^{d-1}x
\end{equation}
is then the phase space of the theory.

\subsection{Electromagnetism}\label{s:EM}
Electromagnetism on a $d$\ndash dimensional lorentzian manifold $(M,g)$ is described by a $1$\ndash form\footnote{More generally, a $U(1)$ connection.}
$A=A_\mu\dd x^\mu$ on $M$ (called the potential) 
 and the action
\[
S_M[A]=\frac14\int_M g^{\mu\nu} g^{\rho\sigma}\,F_{\mu\rho}\,F_{\nu\sigma}
\,\sqrt{|\det g|}\;\dd^dx,
\]
where $F_{\mu\nu}=\de_\mu A_\nu-\de_\nu A_\mu$ is the electromagnetic field. The EL equations are in this case the Maxwell equations
\[
\de_\mu(g^{\mu\nu} g^{\rho\sigma}\,F_{\nu\sigma}\,\sqrt{|\det g|})=0 \quad\forall\rho.
\]

Again we consider $M=\Sigma\times I$, where
$\Sigma$ is space-like and $I=[a,b]$ is a time interval with split metric near the boundary.
In this case, the boundary contribution---the Noether $1$\ndash form---on each copy of $\Sigma$ is
\[
\alpha^\de_\Sigma = \int_\Sigma h^{ij}\,F_{0i}\,\delta A_j\,\sqrt{\det h}\;\dd^{d-1}x,
\]
where $h_{ij}$ is the Euclidean metric on $\Sigma$ obtained by restricting $g$, and $0$ denotes the time component. Again, in the symplectic formalism, we consider $F_{0i}$ and $A_j$ as independent fields. We denote by 
$F^\de_\Sigma$
the space with coordinates $F_{0i}$ and $A_j$ and with symplectic structure
\[
\omega^\de_\Sigma=
\delta\alpha^\de_\Sigma = \int_\Sigma h^{ij}\,\delta F_{0i}\,\delta A_j\,\sqrt{\det h}\;\dd^{d-1}x.
\]
\begin{Rem}\label{r:elF}
If we introduce the electric field $E^j\coloneqq h^{ij}F_{0i}$, then we get the usual formula 
\[
\omega^\de_\Sigma=\int_\Sigma \delta E^j\,\delta A_j \sqrt{\det h}\;\dd^{d-1}x,
\]
which simply says that the electric field and the vector potential are canonically conjugate variables.
\end{Rem}

The crucial point now is that, unlike in the case of the scalar field, the space $F^\de_\Sigma$  is not the space of solutions. In fact, near the boundary,
the EL equations read
\[
\de_0(g^{\rho j}F_{0j}\sqrt{\det h}) + \de_i(h^{ij} g^{\rho\sigma}F_{j\sigma}\sqrt{\det h})=0 \quad\forall\rho.
\]
For $\rho=k$ a space index, we get the evolution equation
\[
h^{kj}\dot F_{0j}\sqrt{\det h}= \de_i(h^{ij} h^{kl}F_{jl}\sqrt{\det h}),
\]
where we denoted the time derivative by a dot (assuming for simplicity that $h$ is time\ndash independent). The equation for $\rho=0$,
\[
\de_i(h^{ij} F_{0j}\sqrt{\det h})=0,
\]
is however not an evolution equation and instead defines a constraint on the variables on $F^\de_\Sigma$. Note that, with the notations in Remark~\ref{r:elF}, this equation corresponds to the Gauss law $\Div E = 0$. We then define $C$ as the set of solutions to the Gauss law or, equivalently, as the common zero set for all scalars $\lambda$ of the functions
\[
J_\lambda\coloneqq -\int_\Sigma \lambda\, \de_i(h^{ij} F_{0j}\sqrt{\det h})\;\dd^{d-1}x = \int_\Sigma \de_i\lambda\, h^{ij} F_{0j}\sqrt{\det h}\;\dd^{d-1}x
\]
The restriction of $\omega^\de_\Sigma$ to $C$ is degenerate. Its kernel is generated by the hamiltonian vector fields $X_\lambda$ for all $J_\lambda$s. We immediately get
\[
X_\lambda(A_i)=\de_i\lambda,\qquad X_\lambda(F_{0j})=0,
\]
which can be interpreted as infinitesimal gauge transformations. The reduced phase space $\underline C$ is then defined as the quotient of the space $C$ of solutions to the Gauss law by gauge transformations.

\subsection{Conclusions}
All the above examples have something in common: the evolution can eventually be described in hamiltonian terms, and the symplectic form is ultimately derived from the Noether $1$\ndash form that arises as boundary term of a variation of the action. 

On the other hand, there are also several differences. In the case of mechanics or of the scalar field this is the end of the story.
In the degenerate case of the length functional, the $2$\ndash form arising from the the differential of the Noether $1$\ndash form is degenerate, so one has to reduce by its kernel in order to get a symplectic form. In the case of electromagnetism (and more generally of Yang--Mills theory), the construction we have presented actually yields a symplectic form on an appropriate space of boundary fields. However, some of the EL equations (the Gauss law) put constraints on them. Imposing the constraints leads however to a degeneracy of the restriction of the symplectic form, so one has now to reduce by its kernel.

In order to put more uniformity to these examples, but also to get ready to deal with more intricated examples, like gravity, we need a more conceptual perspective, which will be developed in the next sections.


\section{Lagrangian field theory}
In this and in the next two sections, we discuss the general method. We encourage the reader to apply each new construction to the examples of Section~\ref{s:prel-exa} (the results will be briefly reviewed in Section~\ref{s-exaafterKT}).

A Lagrangian field theory on a manifold $M$, possibly with boundary, is specified by a space of fields $F_M$ and a local action functional $S_M$. The fields are local on $M$ (sections of some sheaf); typically: 
\begin{itemize}
\item maps from $M$ to some fixed target manifold,
\item sections of some vector (or, more generally, fiber) bundle over $M$,
\item connections on $M$ for some principal bundle.
\end{itemize}
The action functional is of the form $S_M=\int_M L$, where the lagrangian density $L$ is a function of the fields and some of their derivatives (jets) at a point. 
In this note, we assume $M$ to be compact.\footnote{If $M$ is not compact, the fields are assumed to vanish fast enough to ensure convergence of the integral. Alternatively, one considers the Lagrangian as a distribution to be applied to test functions with compact support.}
To start with, we fix no boundary conditions for the fields.

\begin{Rem}
The lagrangian density $L$ may require some additional structure on $M$. 
For example, to define a theory for spinors we need $M$ to allow for spin bundles. To define gravity we need $M$ to allow for lorentzian metrics. For a given theory, we call a manifold that admits the required structures a \textsf{space--time manifold}.
\end{Rem}

The EL equations are encoded in the variation $\delta S_M$ of the action. We think of $\delta$ as the de~Rham differential on $F_M$. For this, we have to assume that $F_M$ has a suitable structure of infinite-dimensional manifold that allows defining differential forms and the de~Rham differential. 

\begin{Rem}
In this note, we prefer to consider smooth fields and give $F_M$ the structure of a Fr\'echet manifold. One may also prescribe a different regularity on the fields and give $F_M$ the structure of a Banach manifold. 
\end{Rem}

\begin{Dig}[A categorical approach]
A different option, which is actually closer to the physical viewpoint but which we will not explore in this note, consists in regarding $F_M$ as an internal hom space in the category of smooth manifolds. Let us explain this in the case of a sigma model where $F_M$ is the space of maps from $M$ to a fixed target manifold $N$ (possibly $\bbR$, as in scalar field theory). We then regard $F_M$ as the functor that to a finite\ndash dimensional manifold $Z$ associates the set of smooth maps from $Z\times M$ to $N$. We think of this as a way of exploring $F_M$ by manifolds $Z$ (it is the notion of Grothendiek's functor of points, but it is also the usual textbook way of parametrizing the fields when computing  functional derivatives). The natural context is that of  ``presheaves"  over the category $\mathbf{Mfld}$ of finite\ndash dimensional smooth manifolds, namely, of functors from $\mathbf{Mfld}^\text{op}$ to the category $\mathbf{Set}$ of sets. A finite\ndash dimensional manifold $X$ can also be regarded as a presheaf $\underline X$, namely, the functor that to a finite\ndash dimensional manifold $Z$ associates the set of smooth maps from $Z$ to $X$. We regard presheaves, e.g., $F_M$ or $\underline X$, as generalized manifolds. A natural transformation among the corresponding functors is regarded as a map between the generalized manifolds. Particularly useful is the generalized manifold $\Omega^k$, namely, the functor that to a finite\ndash dimensional manifold $Z$ associates the set $\Omega^k(Z)$ of smooth $k$\ndash forms on $Z$. A $k$\ndash form on $X$ is now equivalently described as a map from $\underline X$ to $\Omega^k$. One then defines $k$\ndash forms on a generalized manifold, e.g., $F_M$, as maps to $\Omega^k$.
For every $k$, we have the universal de~Rham differential $\delta$, the natural transformation from $\Omega^k$ to $\Omega^{k+1}$, which, when instantiated on a given $Z$, is the usual de~Rham differential $\dd\colon \Omega^k(Z)\to\Omega^{k+1}(Z)$. The action functional $S_M$ is now a map from $F_M$ to $\Omega^0$, namely, the natural tranformation, which, when instantiated on a given $Z$, is the smooth function on $Z$ given by $\int_M L$ (recall that the fields are now parametrized by $Z$). The variational derivative $\delta S_M$ is now just the composition of the universal de~Rham differential $\delta$ with the map $S_M$, a generalized map from $F_M$ to $\Omega^1$, so a $1$\ndash form on $F_M$.
\end{Dig}


\subsection{Euler--Lagrange spaces}
If $M$ has no boundary, solutions to the EL equations are, by definition, the critical points of $S_M$, viz., the fields at which $\delta S_M=0$. If we introduce the Euler--Lagrange $1$\ndash form $el_M\coloneqq\delta S_M$, then the solutions to the EL equations are the zeros of $el_M$. We denote by $EL_M$ the space of solutions.

Note that there is an ambiguity in writing $el_M$, as there was already in writing $S_M$, because of total derivatives, whose integral is zero. If $M$ has nonempty boundary, the different ways are no longer equivalent but differ by boundary terms. 

We take the definition of $S_M$ as fixed (we will return to this later on), but in $el_M$ we only want variations of the fields, and not of their derivatives, to appear. This requires integrating by parts (several times if necessary). Finally, we may write (in a unique way)
\begin{equation}\label{e:deltaSonFM}
\delta S_M = el_M + \alpha_M,
\end{equation}
where $\alpha_M$, also a $1$\ndash form on $F_M$, collects the boundary terms. It follows that 
\begin{equation}
\label{e:alphaM}
\alpha_M=\int_{\de M} a_M,
\end{equation}
 where $a_M$ is a density on the boundary $\de M$ depending on the fields and their jets (including  transversal ones) at a point on $\de M$.

We keep denoting by $EL_M$ the zero locus of $el_M$. In this note, for simplicity, we will assume that $EL_M$ is a submanifold of $F_M$. Note that \eqref{e:deltaSonFM} has the following consequences:
\begin{enumerate}
\item The restriction of $\alpha_M$ to $EL_M$ is an exact $1$\ndash form.
\item If we change the definition of $S_M$ to $\Tilde S_M$ via the addition of total derivatives to the Lagrangian $L_M$, we have $\Tilde S_M = S_M + \phi_M$, where $\phi_M=\int_{\de M} f_M$ collects the boundary terms. Since our convention for $el_M$ is fixed (namely, we write $\delta \Tilde S_M = el_M + \Tilde\alpha_M$), we get
\[
\Tilde\alpha_M=\alpha_M+\delta\phi_M.
\]
\end{enumerate}
Changing the action as in (ii) should be regarded as an equivalence (anyway, the space $EL_M$ is does not change).\footnote{In the quantum theory, what matters is the Gibbs weight $\EE^{\frac\ii\hbar S_M}$, which, in view of this remark, should be considered not as a function but as a section of a line bundle over $F_M$: the change in the definition of $S_M$ is the action of the gauge transformation  $\EE^{\frac\ii\hbar \phi_M}$. {}From this point of view, $\alpha_M$ should be considered as a $1$\ndash form connection.} We then define the $2$\ndash form
\[
\omega_M\coloneqq \delta\alpha_M
\]
which has invariant meaning ($\delta\alpha_M=\delta\Tilde \alpha_M$). Point (i) above then becomes:\footnote{Explicitly, this means that $\omega_M$ evaluated on a solution and contracted with a solution of the linearized equations vanishes. See footnote~\ref{f-restriction}.}
\begin{quote}
The restriction of $\omega_M$ to $EL_M$ is zero.
\end{quote}
Borrowing the terminology of symplectic geometry, we say that $EL_M$ is an isotropic submanifold of $(F_M,\omega_M)$.

\subsection{Reduction}\label{s:reduction}
The $2$\ndash form $\omega_M$ is closed by construction but in general is degenerate, so it is not a symplectic form. We can remedy for this quotienting $F_M$ by the kernel of $\omega_M$, i.e., by the distribution of vector fields $X$ on $F_M$ satisfying $\iota_X\omega=0$.\footnote{Here $\iota$ denotes contraction. Namely, $\iota_X\omega$ is the $1$\ndash form that satisfies $\iota_X\omega(Y)=\omega(X,Y)$ for every vector field $Y$.} 

We will assume that the quotient $F^\de_{\de M}$ has a smooth manifold structure such that the canonical projection $\pi\colon F_M\to F^\de_{\de M}$ is a surjective submersion.\footnote{We first have to assume that this distribution is regular, i.e., that the kernel distribution consists of sections of a vector bundle. Since $\omega$ is closed, this distribution is automatically involutive (viz., $\iota_X\omega=0$ and $\iota_Y\omega=0$ imply $\iota_{[X,Y]}\omega=0$), but, in the Fr\'echet context, the Frobenius theorem does not hold, so we have to assume explicitly that the distribution is integrable. Finally, we have to assume that the leaf space $F^\de_{\de M}$  has smooth structure as in the text. All this turns out to work in the case of usual field theories.}

The reason for the notation using the boundary symbol $\de$ is that, because of \eqref{e:alphaM},  a vector field $X$ that changes the bulk fields preserving their boundary values (including their transversal jets) is necessarily in the kernel of $\omega_M$. As a result, $F^\de_{\de M}$ only depends on ``boundary data.''
We call an element of $F^\de_{\de M}$ a \textsf{boundary field}.

By construction, there is a unique symplectic form $\omega^\de_{\de M}$ on $F^\de_{\de M}$ such that $\pi^*\omega^\de_{\de M}=\omega_M$. 

A central concept is the space of boundary fields that arise from solutions to EL equations:
\begin{equation}
\label{e:LM}
L_M\coloneqq \pi(EL_M).
\end{equation}
Note that $L_M$
is automatically isotropic with respect to $\omega^\de_{\de M}$. We will later require further structure on $L_M$ (viz., that it is a split lagrangian submanifold; see Definition~\ref{d:splitL}).

To make further structures clearer, but also to present a version that still works when $F^\de_{\de M}$ is singular, it is convenient to perform the reduction in stages,  with the first reduction, which we will denote by  $\Tilde F^\de_{\de M}$, always possible. If also $F^\de_{\de M}$ is smooth, we will then have the following commutative diagram of surjective submersions:
\begin{equation}\label{e:d:Fs}
\begin{tikzcd}
  F_M \arrow[rd] \arrow[r, "\Tilde\pi"] \arrow[dr,"\pi"] & \Tilde F^\de_{\de M}  \arrow[d, "p"]  \\
                                                                                & F^\de_{\de M}
\end{tikzcd} 
\end{equation}
We will call points of $\Tilde F^\de_{\de M}$  \textsf{preboundary fields}. We study them in the next section.

\section{Preboundary fields}\label{s:preboundary}
We already observed that a vector field $X$ on $F_M$ that changes the bulk fields preserving their boundary values (including their transversal jets, which encode the inward\ndash pointing derivative data of the fields at the boundary) is necessarily in the kernel of $\omega_M$ (actually, also in that of $\alpha_M$).\footnote{Alternatively, we may consider vector fields that preserve germs of boundary values. Namely, we require that, for every field $\phi$, $X$ evaluated at $\phi$ is compactly supported in the interior of $M$. This construction produces a different version of the intermediate space $\Tilde F^\de_{\de M}$, but of course the final reduction $F^\de_{\de M}$ is the same.} 
 We now consider the distribution consisting of only such vector fields and denote by $\Tilde F^\de_{\de M}$ its leaf space and by 
\[
\Tilde\pi\colon F_M\to \Tilde F^\de_{\de M}
\]
the canonical projection.
This is a space of fields on $\de M$ (with the transversal jets now being new  fields). As already mentioned, we call an element of $\Tilde F^\de_{\de M}$ a \textsf{preboundary field}.

Note that $\Tilde\pi$ is a surjective submersion (corresponding to restricting fields and their transversal jets to the boundary) and that there is a unique closed local $2$\ndash form $\Tilde\omega^\de_{\de M}$ on $\Tilde F^\de_{\de M}$ such that $\Tilde\pi^*\Tilde\omega^\de_{\de M}=\omega_M$.  

The $2$\ndash form $\Tilde\omega^\de_{\de M}$ is in general still degenerate, so a further reduction will be needed, leading again to $F^\de_{\de M}$ with canonical projection
$p\colon F_M\to F^\de_{\de M}$. 

There is also a unique local $1$\ndash form $\Tilde\alpha^\de_{\de M}$ on $\Tilde F^\de_{\de M}$ such that $\Tilde\pi^*\Tilde\alpha^\de_{\de M}=\alpha_M$. Moreover, we have $\Tilde\omega^\de_{\de M}=\delta \Tilde\alpha^\de_{\de M}$ and 
\begin{equation}\label{e:deltaStilde}
\delta S_M = el_M + \Tilde\pi^*\Tilde\alpha^\de_{\de M}.
\end{equation}  

We now define the analogue of \eqref{e:LM}, namely, the space of preboundary fields that can be extended to a, not necessarily unique, solution to the EL equations: 
\[
\Tilde L_M\coloneqq \Tilde\pi(EL_M).
\]
Note that $\Tilde L_M$
is automatically isotropic with respect to $\Tilde\omega^\de_{\de M}$ and that $L_M=p(\Tilde L_M)$.
It is convenient to assume that $\Tilde L_M$ is a submanifold of $\Tilde F^\de_{\de M}$. 

\subsection{Composition  of evolution relations}
Note that $\Tilde F^\de_{\Sigma}$ is actually defined for every manifold $\Sigma$ that can arise as a boundary component of a space--time manifold $M$.
Moreover, if we denote by $\{\de_s M\}_{s=1,\dots,k}$ the connected components of $\de M$, we have\footnote{Here $\prod$ denotes the cartesian product of the given spaces. In the second formula, the pulback of $\Tilde\alpha^\de_{\de_s M}$ by the projection to the corresponding factor is understood.}
\[
\Tilde F^\de_{\de M}=\prod_{s=1}^k \Tilde F^\de_{\de_s M}
\]
and
\[
\Tilde\alpha^\de_{\de M} = \sum_{s=1}^k \Tilde\alpha^\de_{\de_s M}.
\]


Suppose now we cut the manifold $M$ along a hypersurface $\Sigma$ into two pieces $M_1$ and $M_2$, so we can recover $M$ as the gluing of $M_1$ and $M_2$ along their common boundary component $\Sigma$ (we assume here that $\Sigma$ does not cut the boundary of $M$). We have
\[
\de M = \de_1 M\sqcup\de_2M,\quad
\de M_1 = \de_1 M\sqcup\Sigma,\quad
\de M_2 = \de_2 M\sqcup\Sigma,
\]
where $\de_i M$  denotes the component of the boundary of $M$ that is also a component of the boundary of $M_i$.

A solution to the EL equations on $M$ is then uniquely given by a pair of solutions to the EL equations on $M_1$ and $M_2$ that match on $\Sigma$. More precisely, denote by $\Tilde\pi_i\colon F_{M_i}\to \Tilde F^\de_{\Sigma}$, $i=1,2$, the map that corresponds to restricting fields on $M_i$ and their transversal jets to the boundary component $\Sigma$. We then have the fiber product formula
\begin{equation}\label{e:compEL}
EL_M = EL_{M_1}\times_{\Tilde F^\de_{\Sigma}}EL_{M_2}
=\{(\phi_1,\phi_2)\in EL_{M_1}\times EL_{M_2} : \Tilde\pi_1(\phi_1)=\Tilde\pi_2(\phi_2)\}.
\end{equation}
Analogously, we have
\begin{equation}\label{e:compLtilde}
\Tilde L_M = \{(\psi_1,\psi_2)\in \Tilde L_{M_1}\times \Tilde L_{M_2} : \pi_{\Sigma,1} (\psi_1)=\pi_{\Sigma,2}(\psi_2)\},
\end{equation}
where $\pi_{\Sigma,i}$ is the canonical projection from $\Tilde F^\de_{\de M_i}$ to $\Tilde F^\de_{\Sigma}$. 

It is more instructive to view the above formulae using the language of relations and correspndences.
\begin{Def}
A \textsf{relation} from the set $A$ to the set $B$ is a subset of $A\times B$. If $R_{AB}$ is a relation from $A$ to $B$ and 
$R_{BC}$ is a relation from $B$ to $C$, one defines the composition
\[
R_{BC}\circ R_{AB} \coloneqq \{(a,c)\in A\times C : \exists b\in B\ (a,b)\in R_{AB}\ (b,c)\in R_{BC}\}
\]
as a relation from $A$ to $C$.\footnote{A particular example of a relation is the graph of a map. The composition of  graphs, as relations, is the same as the graph of the usual composition of maps.} A \textsf{correspondence} from the set $A$ to the set $B$ is a map $C\to A\times B$ where $C$ is some set. The composition of correspondences is defined as their fibered product. The image of a correspondence is a relation, and the composition of such images, as relations, is  the image of the composition of the correspondences.
\end{Def}
In this setting we may
view $\Tilde L_{M_1}$ ($EL_{M_1}$) as a relation (correspondence) from $\Tilde F^\de_{\de_1 M}$
to $\Tilde F^\de_{\Sigma}$ and $\Tilde L_{M_2}$ ($EL_{M_2}$) as a relation  (correspondence) from  $\Tilde F^\de_{\Sigma}$ to $\Tilde F^\de_{\de_2 M}$. In this language, \eqref{e:compLtilde} is the composition of relations, whereas \eqref{e:compEL} is the  composition of correspondences. 

For this reason, given a space--time manifold $M$ and a decomposition of its boundary into two components,
we call $\Tilde L_M$ ($EL_M$) the \textsf{evolution relation} (\textsf{evolution correspondence}).

\begin{Rem}[Evolutionary flow]
A particular case is $M=\Sigma\times [t_1,t_3]$, with $\Sigma$ having empty boundary and $t_1<t_3$. 
We now pick $t_2\in(t_1,t_3)$ and cut $M$ along $\Sigma\times\{t_2\}$. In this case, $M_i=\Sigma\times [t_i,t_{i+1}]$,
and we have 
\[
\Tilde L_{\Sigma\times [t_1,t_3]} =  \Tilde L_{\Sigma\times [t_2,t_3]} \circ \Tilde L_{\Sigma\times [t_1,t_2]}. 
\]
This law may be thought of as the general version of time evolution (we will see that in many cases there is a further reduction in which the evolution relation actually becomes the graph of a flow).
\end{Rem}


\begin{Rem}[Relative structures]\label{r:relstr}
Suppose again we cut the manifold $M$ along a hypersurface $\Sigma$ into two pieces $M_1$ and $M_2$.
 Since $S_M= S_{M_1}+S_{M_2}$, equation \eqref{e:deltaStilde} for $M$, $M_1$, and $M_2$ implies
\[
\Tilde\alpha^\de_{\Sigma,1}=-\Tilde\alpha^\de_{\Sigma,2},
\]
where $\Tilde\alpha^\de_{\Sigma,i}$ denotes the $1$\ndash form on $\Tilde F^\de_{\Sigma}$ viewed as a boundary component of $M_i$. Therefore, with analogue notation, we have
\[
\Tilde\omega^\de_{\Sigma,1}=-\Tilde\omega^\de_{\Sigma,2}.
\]
This shows that choosing a boundary component $\Sigma$ as the target or the source of the evolution relation 
yields opposite sign to the associated $2$\ndash form.\footnote{If our theory is defined on oriented space--time manifolds and the Lagrangian is viewed as a top form, then we view $a_M$ in \eqref{e:alphaM} also as a top form (now on the boundary). In this setting, $\alpha_M$ depends on a choice of orientation, and we choose the orientation of $\Sigma$ in opposite ways depending on its being viewed as the source or the target of the evolution relation.}
\end{Rem}

\subsection{Cauchy data}
Let $\Sigma$ be a manifold that can appear as a boundary of a space--time manifold.
We then consider the theory on $M_\epsilon=\Sigma\times[0,\epsilon]$ with $\epsilon$ a positive real number.
We view $\Tilde L_{M_\epsilon}$ as a relation from $\Tilde F^\de_{\Sigma}=\Tilde F^\de_{\Sigma\times\{0\}}$ to $\Tilde F^\de_{\Sigma\times\{\epsilon\}}$. We define
\[
\Tilde C_\Sigma \coloneqq \{c\in \Tilde F^\de_{\Sigma} : \exists \epsilon > 0\ \exists u\in F^\de_{\Sigma\times\{\epsilon\}}\ (c,u)\in \Tilde L_{M_\epsilon}\}.
\]
The space $\Tilde C_\Sigma$ consists of all the preboundary fields that can be extended to solutions on some cylinder. Note that in general $\Tilde C_\Sigma$ is not the whole of $\Tilde F^\de_{\Sigma}$ (in the examples we will see that $\Tilde C_\Sigma$ is determined by the constraints of the theory, viz., by those EL equations that do not describe an evolution in the transverse direction). We call $\Tilde C_\Sigma$ the \textsf{space of Cauchy data}. 

More generally, we denote by $\Tilde C_{\de M}$ the product of the spaces of Cauchy data associated to each boundary component of $M$, viewed as a subset of $\Tilde F^\de_{\de M}$.

We may restrict the closed $2$\ndash form $\Tilde\omega^\de_{\Sigma}$ to $\Tilde C_\Sigma$ and denote this restriction by $\Tilde\omega^C_{\Sigma}$. In general, it will be degenerate. In some good cases, $\Tilde C_\Sigma$ is a submanifold and
the leaf space $\underline{C_\Sigma}$ of the distribution given by the vector fields in the kernel of $\Tilde\omega^C_{\Sigma}$ has a smooth structure such the the canonical projection $\Tilde p_C\colon\Tilde C_\Sigma\to\underline{C_\Sigma}$ is a surjective submersion. In this case, there is a unique symplectic structure $\underline\omega^C_\Sigma$ on $\underline{C_\Sigma}$ satisfying $\Tilde\omega^C_{\Sigma}=\Tilde p_C^*\underline\omega^C_\Sigma$. We call the symplectic manifold $(\underline{C_\Sigma},\underline\omega^C_\Sigma)$ the
\textsf{reduced phase space} associated to $\Sigma$. If this constuction works for every boundary component of $M$, we write $\underline{C_{\de M}}$ for the product of the symplectic manifolds associated to each component. The evolution relation $\Tilde L_M$ may be intersected with $\Tilde C_{\de M}$ and projected to $\underline{C_{\de M}}$ giving rise to the reduced evolution relation $\underline L_M$. This describes the correct evolution of the system. (In some particularly good cases, $\underline{L}_{\Sigma\times[a,b]} $ will be the graph of a hamiltonian flow.)


\section{Boundary fields}
We now return to the full reduction $F_M\to F^\de_{\de M}$ mentioned in Section~\ref{s:reduction}; see the diagram 
\eqref{e:d:Fs} for reference.

We assume here that the quotient $F^\de_{\de M}$ has a smooth manifold structure such that the canonical projection $\pi\colon F_M\to F^\de_{\de M}$ is a surjective submersion and recall that an element of $F^\de_{\de M}$ is called a \textsf{boundary field}.

The excursus through the space of preboundary fields in Section~\ref{s:preboundary} and the fact that we can obtain $F^\de_{\de M}$ from $\Tilde F^\de_{\de M}$ by a further reduction allow us to see more structure.

The first consequence is that, for every manifold $\Sigma$ that can arise as a boundary component of a space--time manifold $M$, we may define $F^\de_{\Sigma}\coloneqq p(\Tilde F^\de_{\Sigma})$. This implies that, 
if we denote by $\{\de_s M\}_{s=1,\dots,k}$ the connected components of $\de M$, we have
\[
F^\de_{\de M}=\prod_{s=1}^k F^\de_{\de_s M}
\]
and
\[
\omega^\de_{\de M} = \sum_{s=1}^k \omega^\de_{\de_s M}.
\]

The second consequence is that $L_M=\pi(EL_M)=p(\Tilde L_M)$ is also a relation, which we keep calling the \textsf{evolution relation}, once we choose a decomposition of $\de M$ into two components. Moreover, the evolution relation of a manifold cut along a hypersurface is the composition of the evolution relations of its parts.
Remark~\ref{r:relstr} now implies that 
\[
\omega^\de_{\Sigma,1}=-\omega^\de_{\Sigma,2},
\]
so $F^\de_\Sigma$ has opposite symplectic structure if it is regarded as target instead of source space.

\begin{Dig}[Composition of isotropic relations]
Following \cite{Wei71,Wei10},
we can give a better description of the composition of evolution relations, which we now assume to be submanifolds, in terms of symplectic geometry. 
Observe that $L_M$ can be obtained as follows: first, we consider $L_{M_1}\times L_{M_2}$ as a submanifold of $ F^\de_{\de M_1}\times F^\de_{\de M_2}=F^\de_{\de_1 M}\times F^\de_{\Sigma}\times F^\de_{\Sigma}\times F^\de_{\de_2 M}$; then we intersect it with $F^\de_{\de_1 M}\times  \Delta_{F^\de_{\Sigma}}\times F^\de_{\de_2 M}$, where $\Delta_{F^\de_{\Sigma}}$ denotes the diagonal in $F^\de_{\Sigma}\times F^\de_{\Sigma}$; finally we project to $F^\de_{\de M}$ by the map $\pi_\Delta\colon F^\de_{\de_1 M}\times  \Delta_{F^\de_{\Sigma}}\times F^\de_{\de_2 M}\to F^\de_{\de M}$ that forgets the middle factor. Therefore, we have
\[
L_M = \pi_\Delta\left(L_{M_1}\times L_{M_2} \cap  F^\de_{\de_1 M}\times  \Delta_{F^\de_{\Sigma}}\times F^\de_{\de_2 M}\right).
\]
The point of this construction is that the map $\pi_\Delta$ performs the symplectic reduction of the restriction of $\omega^\de_{F^\de_{\de M_1}\times F^\de_{\de M_2}}$ to $F^\de_{\de_1 M}\times  \Delta_{F^\de_{\Sigma}}\times F^\de_{\de_2 M}$. The crucial point here is that $\Delta_{F^\de_{\Sigma}}$ is a lagrangian submanifold of
$F^\de_{\Sigma}\times F^\de_{\Sigma}$ when the two factors are given, as is the case here, opposite symplectic structures. 
\end{Dig}

\subsection{Boundary conditions}
In this note we assume that there is a local $1$\ndash form $\alpha^\de_{\de M}$ such that $\pi^*\alpha^\de_{\de M}=\alpha_M$.\footnote{This condition actually fails in some interesting examples (like a charged particle in a magnetic field or the WZW model). 
What still happens in these examples however is that $\alpha^\de_{\de M}$ exists as a $1$\ndash form connection instead of a global $1$\ndash form.}
It follows that
\begin{enumerate}
  \item $\alpha^\de_{\de M}$ is uniquely determined,
  \item $\omega^\de_{\de M}=\delta \alpha^\de_{\de M}$, and
  \item equation \eqref{e:deltaSonFM} yields
\begin{equation}\label{e:deltaSondeM}
\delta S_M = el_M + \pi^*\alpha^\de_{\de M}.
\end{equation}  
  \end{enumerate}


In order to select solutions to the EL equations that could also be regarded as critical points of $S_M$, one may then choose a submanifold $B$ of $F^\de_{\de M}$ with the following two properties:
\begin{enumerate}
\item The restriction of $\alpha^\de_{\de M}$ to $B$ vanishes, so $\delta S_M=el_M$ on $\pi^{-1}(B)$.
\item $B$ intersects $L_M$ transversally, so each solution is isolated.
\end{enumerate}
We will see that such submanifolds  $B$ are actually related to possible boundary conditions of the theory.

The first condition implies that $B$ should also be isotropic. The second condition in particular entails that $T_u B$ and $T_u L_M$ are complementary subspaces of $T_u F^\de_{\de M}$, for every intersection point $u$ (i.e., the boundary value  in $B$ of a solution). Note that $\omega^\de_{\de M}$ at $u$ defines a 
symplectic form
on the vector space $T_uF^\de_{\de M}$, and that $T_uB$ and $T_u L_M$ are isotropic with respect to it. We say that $T_uB$ and $T_uL_M$ are split lagrangian subspaces, according to the following
\begin{Def}\label{d:splitL}\cite{CC21}
An isotropic subspace of a symplectic vector space is called split lagrangian if it admits an isotropic complement.
\end{Def}
The motivation for this terminology is that an isotropic subspace of a finite\ndash dimensional symplectic space 
is split lagrangian precisely when it is half\ndash dimensional, so lagrangian according to the usual definition. In the infinite\ndash dimensional symplectic case, there are several definitions, of different strength, of what lagrangian should be, split lagrangian 
being the strongest.\footnote{Note that, if $L$ is split lagrangian in a symplectic vector space $V$, then its symplectic orthogonal $L^\perp$ is equal to $L$, so the symplectic reduction of $L$ is a point. The condition $L=L^\perp$ is usually taken to define lagrangian subspaces. In particular, split lagrangian implies lagrangian.}

If we want that every point $u$ of $ L_M$ could be obtained by such an intersection, we have to require that $T_u L_M$ should be split lagrangian for every $u\in  L_M$: we say, in this case, that $ L_M$ is a split lagrangian submanifold.\footnote{The submanifolds $B$ that select isolated solutions have also to be split lagrangian, but in principle it is enough to require this in a neighborhood of the intersection points.}

By construction,
$L_M$ is always isotropic. 
On the other hand, the condition that $L_M$ should be split lagrangian is required for the theory to behave well when boundary conditions are imposed and does not necessarily hold in general. Also note that this may depend on the choice of space--time manifold $M$. For a Lagrangian theory to be good, we have to require that there is a class of manifolds with boundary  for which the condition holds (this class should at least contain manifolds of the form $\Sigma\times I$, with $I$ a time interval and $\Sigma$ belonging to a physically interesting class, e.g., being space-like). The theories encountered in physics (scalar or spinor field theories and Yang--Mills theories, pure or coupled to the former) are good in this sense.

In the case of gravity, we will see that one has to be a bit more careful in the definition of $F_M$ in order for the theory to be well behaved on a large class of space--time manifolds with boundary: namely, one has to impose some open condition for the fields in $F_M$ on the boundary of $M$.\footnote{We call this the boundary metric nondegenerate PC theory in Section~\ref{s-boundary metric nondegenerate PC theory}.} 
We do not view this as a boundary condition (such as the one imposed by the intersection  with an isotropic $B$) which would instead correspond to a closed condition.

\begin{Rem}[Coisotropic Cauchy spaces]
One can show that, if $L_{\Sigma\times[0,\epsilon]}$ is lagrangian for all $\epsilon>0$, then $C_\Sigma$ is coisotropic (i.e., the symplectic orthogonal of its tangent bundle is contained in the  tangent bundle itself). In this case, $C_\Sigma$ is locally given by first\ndash class constraints (using Dirac's terminology). Moreover, the reduced phase space $\underline C_\Sigma$ is obtained modding out by the hamiltonian vector fields of the constraints (which one can regard as gauge transformations).
This is the typical case we will analyze in the rest of this note.
\end{Rem}

\section{Examples}\label{s-exaafterKT}
We now briefly return to the examples discussed in Section~\ref{s:prel-exa}, in light of the general KT method.

In the case of mechanics, with target $\bbR$ as in Section~\ref{s:mechanics}, we have $\underline C=C=F^\partial =T\bbR$, with coordinates interpreted as initial position $q$ and initial velocity $v$. The symplectic form is $m\,\delta v\,\delta q$. If we set $p\coloneqq mv$, we identify these spaces with $T^*\bbR$ with canonical symplectic form $\delta p\,\delta q$. The evolution relation $\underline L_{[a,b]}\subset \underline C\times \underline C=T^*\bbR\times T^*\bbR$ is the graph of the hamiltonian flow from time $a$ to time $b$.

In the degenerate example of Section~\ref{s:degenerate}, we have $\underline C=C=F^\partial =T^*S^{n-1}$ with canonical symplectic structure. The evolution relation $\underline L_{[a,b]}$ is just the graph of the identity.

In the case of the scalar field of Section~\ref{s:scalar}, we have $\underline C_\Sigma=C_\Sigma=F^\partial_\Sigma =C^\infty(\Sigma)\oplus C^\infty(\Sigma)\ni(\phi,\phi_0)$ and symplectic form as in \eqref{e:symplscalar}. 
The evolution relation $\underline L_{\Sigma\times [a,b]}$ consists of pairs $((\phi^a,\phi^a_0),(\phi^b,\phi^b_0))$ with $(\phi^b,\phi^b_0)$ the evaluation at time $b$ of the solution $\phi$ and its time derivative $\de_0\phi$ of the Cauchy problem with initial condition $(\phi^a,\phi^a_0)$ at time $a$.

In the case of electromagnetism of Section~\ref{s:EM}, we have $F^\partial_\Sigma=\frX(\Sigma)\oplus\frX(\Sigma)\ni(A,E)$ and symplectic form as in Remark~\ref{r:elF}.\footnote{$\frX(\Sigma)$ denotes the vector fields on $\Sigma$.}
The Cauchy space $C_\Sigma$ consists of the pairs $(A,E)$ with $E$ satisfying the Gauss law $\Div E=0$. The reduced phase space $\underline C_\Sigma$ is the quotient of $C_\Sigma$ by gauge transformations $A_i\mapsto A_i+\de_i\lambda$. The evolution relation $\underline L_{\Sigma\times [a,b]}$ consists of pairs $(([A]^a,E^a),([A]^b,E^b))$ with $([A]^b,E^b)$ the evaluation at time $b$ of the solution of the Cauchy problem for the time\ndash dependent Maxwell equations with initial condition $([A]^a,E^a)$ at time $a$.\footnote{$[A]$ denotes the gauge equivalence class of $A$.}


\section{Gravity in the coframe formalism}\label{s-boundary metric nondegenerate PC theory}
We now come to the example of the reduced phase space we want to discuss in this note: gravity. In its usual formulation, gravity is described by a Lorentzian metric $g$ and by the Einstein--Hilbert action. The construction of the reduced phase space turns out to be more convenient and elegant in its formulation via (co)frames (using ideas mainly stemming from Cartan and Weyl), usually called the Palatini--Cartan (PC) theory or the first\ndash order formulation of gravity (see, e.g., the reviews \cite{HHKN76} and \cite{Tecch20} and references therein).
We will mainly follow \cite{CCS20}.\footnote{For the treatment with Dirac's method, see, e.g., \cite{Ash87,Bar00}. For the covariant phase space formalism, see \cite{OS19,BMBV21}.}

The first dynamical field in this formulation is a smooth field of frames on our space--time manifold $M$. The dynamical metric $g$ is recovered using the frame starting from some reference metric (whose choice is irrelevant). 
A frame field may also be viewed as an isomorphism between a reference vector bundle $\calV$ on $M$ and the tangent bundle $TM$. It turns out that the computations are handier if one uses the inverse of a frame field, which we call a coframe (field). It turns out that to construct the action one also needs a second dynamical field, a connection for the orthogonal frame bundle.
This construction renders gravity closer in spirit to gauge theories.
More precisely, given a $d$\ndash dimensional manifold $M$, the data are
\begin{itemize}
\item a vector bundle $\calV$ over $M$ isomorphic to $TM$ (over the identity map on $M$), and
\item a smooth family $\eta$  of inner products with signature $(1,d-1)$ on the fibers of $\calV$.\footnote{Note that this is possible if and only if $M$ admits a Lorentzian structure.}
\end{itemize}

The first dynamical field is an isomorphism $e\colon TM\to \calV$ (over the identity map on $M$), called the coframe. The dynamical gravity field $g$ is recovered pulling $\eta$ back by $e$. Namely, if we pick coordinates on $M$ (we use Greek indices for them) and a local basis for the fiber of $\calV$ (we use the first lower case letters of the Latin alphabet for its components), then the coframe $e$ is locally specified by its components $e_\mu^a$ (an invertible $d\times d$ matrix at each point of $M$).\footnote{We denote by $\bar e\colon \calV\to TM$, with components $\bar e^\mu_a$, the inverse of $e$. At each point $x\in M$, for each $a$ we have a tangent vactor $\bar e_a(x)$ with components $(\bar e^1_a(x),\dots,\bar e^d_a(x))$. The vectors $\bar e_1(x), \dots,\bar e_d(x)$ are linearly independent, so they yield a frame for the tangent space $T_xM$. For this reason, $\bar e$ is called a frame and $e$ a coframe.}
 If we denote the components of $\eta$ by $\eta_{ab}$, we get the components of the space--time metric $g$ as
\begin{equation}\label{e:g}
g_{\mu\nu}=e_\mu^a\,e_\nu^b\,\eta_{ab}.
\end{equation}
Note that $g$ has only $d(d+1)/2$ independent components, so the $d^2$ components of $e$ are $d(d-1)/2$ too many. Essentially, we have to consider the internal orthogonal group to mod out the extra components. 

In particular, we have to introduce a connection $A$ to compute the covariant derivative $\bD$ of sections of $\calV$. If we denote by $\phi^a$ the components of a section $\phi$ of $\calV$, we have
\[
(\bD_\mu\phi)^a=\de_\mu\phi^a+A^a_{\mu b}\phi^b.
\]
In particular we want this connection to be compatible with the metric $\eta$ or, equivalently, its inverse, whose components we denote by $\eta^{ab}$. Namely,
 \[
0= (\bD_\mu\eta)^{ab}=\de_\mu\eta^{ab}+A^a_{\mu c}\eta^{cb}+A^b_{\mu c}\eta^{ac}.
 \]
It is convenient to decompose
\[
A^a_{\mu b}=\eta_{bc}(\alpha_\mu^{ac}+\omega_\mu^{ac})
\]
with $\alpha$ symmetric and $\omega$ antisymmetric in the upper indices. The above equations then read
\[
0=\de_\mu\eta^{ab}+\alpha_\mu^{ab}+\alpha_\mu^{ba},
\]
which can be solved as
\[
\alpha_\mu^{ab}=-\frac12\de_\mu\eta^{ab}.
\]
Note that there are no further conditions on $\omega$. The covariant derivative can then be written as
\[
(\bD_\mu\phi)^a=(D_\mu\phi)^a+\omega_\mu^{ab}\eta_{bc}\phi^c
\]
with
\[
(D_\mu\phi)^a=\de_\mu\phi^a-\frac12\de_\mu\eta^{ab}\,\eta_{bc}\phi^c.
\]
The second field of the theory is $\omega$. If we pick a reference $\omega_0$, then $\omega-\omega_0$ is a $1$\ndash form taking values in $\Lambda^2\calV$.

It is convenient to introduce the following piece of terminology. A $k$\ndash form on $M$ taking values in $\Lambda^l\calV$ will be called a $(k,l)$\ndash form or a form of type $(k,l)$. In particular, $e$ is a $(1,1)$\ndash form and $\omega-\omega_0$ is a $(1,2)$\ndash form. The operator $\bD$ acts on $(k,l)$\ndash forms producing $(k+1,l)$\ndash forms by skew-symmetrizing the lower space--time indices. We call this operation the covariant differential and denote it $\dd_\omega$. 

The curvature $2$\ndash form $F_\omega$ is the $(2,2)$\ndash form
defined via $\dd_\omega^2\phi = F_\omega\cdot\phi$, for any section $\phi$ of $\calV$, with
$(F_\omega\cdot\phi)^a=\eta_{bc}F_\omega^{ac}\phi^b$ in local coordinates.

{}From now on, we focus on the four\ndash dimensional case, $d=4$, with cosmological constant $\Lambda$. For a space--time manifold $M$, admitting $\calV$ as above, we define the action functional (the Palatini--Cartan action)
\[
S_M[e,\omega]\coloneqq\int_M\left(
\frac12e^2F_\omega+\frac\Lambda{24}e^4
\right),
\]
where $e^k$ denotes the $(k,k)$\ndash form obtained by taking the wedge product of $k$ factors $e$. Here the wedge product is done both with respect to the form and to the $\calV$\ndash components.\footnote{In local coordinates, the integrand is then given, up to a normalization factor, by
\[
\epsilon_{abcd}\epsilon^{\mu\nu\rho\sigma}
\left(\frac12 e^a_\mu e^b_\nu(F_\omega)^{cd}_{\rho\sigma}+\frac\Lambda{24}e^a_\mu e^b_\nu e^c_\rho e^d_\sigma\right).
\]}
It can be checked \cite{CCS20} that $(4,4)$\ndash forms, in particular the integrand in the action, are canonically identified with densities. By taking a variation of $S_M$, as in \eqref{e:deltaSonFM}, we get
\[
\alpha_M=\int_{\de M}\frac12 e^2\,\delta\omega
\]
and the EL equations
\begin{subequations}\label{e:EL}
\begin{align}
eF_\omega+\frac\Lambda6e^3&=0,\label{se:ELeF}\\
e\dd_\omega e &=0.\label{se:ELede}\\
\intertext{One can check that the second equation is actually equivalent to}
\dd_\omega e& =0,\label{se:ELde}
\end{align} 
\end{subequations}
as a consequence of the fact that $e$ is nondegenerate. 
Moreover, for a given $e$, there is a unique connection $\omega(e)$ satisfying this equation. The connection $\omega(e)$ corresponds, via $e$, to the Levi-Civita connection for $g$ as in \eqref{e:g}. Finally, the equation $eF_{\omega(e)}+\frac\Lambda6e^3=0$ corresponds to the Einstein equation, with cosmological constant $\Lambda$, for $g$.

\begin{Not}\label{n:eomegade}
{}From now on we will put a tilde on the name of the bulk fields, so we will write $\Tilde e=\Tilde e_\mu\dd x^\mu$ and $\Tilde\omega=\Tilde\omega_\mu\dd x^\mu$ (with $\mu$ running from $1$ to $4$). We will reserve the notation without the tilde for the fields on the boundary: $e=e_i\dd x^i$ and $\omega=\omega_i\dd x^i$ (with $i$ running from $1$ to $3$ and $x\in\de M$). Note that the fields on the boundary take values in the restriction $\calV|_{\de M}$ of $\calV$ to the boundary. The nondegeneracy condition now says that the three components $e_1(x)$, $e_2(x)$, and $e_3(x)$ are linearly independent in $\calV_x$ for every $x\in\de M$.
\end{Not}

{}From the form of $\alpha_M$, it might seem that the space of boundary fields consists of $e$ and $\omega$ as 1\ndash forms on the boundary (taking values in $\calV|_{\de M}$ and $\Lambda^2\calV |_{\de M} $, respectively). 
 Let us denote this space by $\Hat F^\de_{\de M}$.
It turns out however that the $2$\ndash form 
\begin{equation}\label{e:preomegaPC}
\omega_M=\delta\alpha_M=\int_{\de M} e\,\delta e\,\delta\omega
\end{equation}
is degenerate, so some reduction is still needed in order to get the space $F^\de_{\de M}$ of boundary fields.
More precisely, if $X$ is a vector field on $\Hat F^\de_{\de M}$ with component $u$ along $e$ and $v$ along $\omega$, we get that $\iota_X\omega_M=0$ if{f} $eu=0$ and $ev=0$. One can check that nondegeneracy of $e$ implies $u=0$. However, $v$ may be nonzero.\footnote{One can check that, for each $x\in\de M$, the space of $v(x)$s satisfying $e(x)\,v(x)=0$ is six\ndash dimensional.} In conclusion, the space $F^\de_{\de M}$ of boundary fields consists of coframes $e$ on the boundary and equivalence classes $[\omega]$ of connections under the equivalence $\omega\sim\omega+v$ with $ev=0$. We denote by $\Hat p$ the canonical projection $\Hat F^\de_{\de M}\to F^\de_{\de M}$.

We now want to discuss the Cauchy data. We then assume $M=\Sigma\times[0,\epsilon]$ and consider the first boundary component $\Sigma=\Sigma\times\{0\}$. We have to split the EL equations \eqref{se:ELeF} and \eqref{se:ELde} into constraints and evolution equations. The constraints are just their restrictions, as differential forms, to the boundary;
this way, using Notation~\ref{n:eomegade}, we can define the space of Cauchy data $\Hat C_\Sigma$ inside $\Hat F^\de_{\Sigma}$ as
\[
\Hat C_\Sigma=\left\{(e,\omega)\in \Hat F^\de_{\Sigma} : eF_\omega+\frac\Lambda6e^3 = 0,\ \dd_\omega e=0\right\}.
\]
 The space of Cauchy data $C_\Sigma$ is then obtained as $\Hat p(\Hat C_\Sigma)$. 

Proceeding this way is however a bit difficult, since the constraint  function $\dd_\omega e$ is not invariant under the equivalence relation.\footnote{The other constraint function, $eF_\omega+\frac\Lambda6e^3$, is instead invariant upon using $\dd_\omega e=0$. In fact, upon sending $\omega$ to $\omega+v$, it becomes $e\dd_\omega v$ which is equal to $\dd_\omega(ev)-\dd_\omega e\,v$. The first term vanishes by the condition on $v$ and the second by the other constraint.} 
It would indeed be simpler to consider instead the space
\begin{equation}\label{e:C'}
\Hat C'_\Sigma=\left\{(e,\omega)\in \Hat F^\de_{\Sigma} : eF_\omega+\frac\Lambda6e^3 = 0,\ e\dd_\omega e=0\right\}
\supsetneq \Hat C_\Sigma.
\end{equation}
Note that we have replaced the condition $\dd_\omega e=0$ with the strictly weaker condition
\begin{equation}\label{e:edecon}
e\dd_\omega e=0,
\end{equation}
and it is not difficult to show that this is invariant under the equivalence relation.
It turns out that replacing $\Hat C_\Sigma$ with $\Hat C'_\Sigma$ is indeed possible under a simple assumption.
\begin{Def}
We say that $e$ is metric nondegenerate if $g^\de_{ij}\coloneqq e_i^a\,e_i^b\,\eta_{ab}$ is nondegenerate.
\end{Def}
In order to get to this situation, we have to define the space of fields $F_M$ accordingly. Namely, we have to put the condition that the coframe $\Tilde e$ becomes metric nondegenerate when restricted to the boundary. 
A stricter condition would be that the coframe $\Tilde e$ produces a riemannian metric on $\de M$ (i.e., $\de M$ is space-like). These are open conditions on the space of fields, so they are part of the definition of the theory.\footnote{Reference \cite{CCT20} studies the case when $\de M$ is light-like. For the treatment of this case with Dirac's method, see \cite{AS14}.} We call the PC theory with coframe field required to be metric nondegenerate on the boundary the \textsf{boundary metric nondegenerate PC theory}.
We then have \cite{CS17,CCS20} the
\begin{Thm}\label{t:emnd}
In the boundary metric nondegenerate PC theory,
 $C_\Sigma=\Hat p(\Hat C'_\Sigma)$.
\end{Thm}
This means that in studying $C_\Sigma$ we can instead use $\Hat C'_\Sigma$ and work up to the equivalence relation. This turns out to be very convenient, e.g., in proving that $C_\Sigma$ is coisotropic (see Theorem~\ref{t:PCCcoiso} below).

\begin{proof}
Instead of restricting \eqref{se:ELde} to the boundary, 
it is more convenient to go back to \eqref{se:ELede} and separate the constraints and the evolution equations there.
It is clear that \eqref{e:edecon}
is a constraint.
 On the other hand, there are more constraints hidden in the 
transversal part of the equation. We get, using Notation~\ref{n:eomegade} and denoting the transversal index by $n$,
\begin{equation}\label{e:ELeden}
\Tilde e_n(\dd_{\Tilde\omega}\Tilde e)+\Tilde e(\dd_{\Tilde\omega}\Tilde e)_n=0.
\end{equation}
The transversal derivative of $\Tilde e$ in the last term shows that this equation also describes evolution. The problem is that it is not in normal form. In particular, it might happen that the first term is not of the form $\Tilde e$ times something, which would prevent the evolution equation to be solved. The condition for this not to happen is also a constraint.

{}From now on we focus for simplicity on the case when the space of fields $F_M$ is defined by the condition that the coframe $\Tilde e$ produces a riemannian metric $g^\de$ on $\de M$ (i.e., $\de M$ is space-like).\footnote{The general case is treated in \cite{CCS20}.}
This condition in particular implies that the restriction $\calV|_{\de M}$ of $\calV$ admits a global time-like section $\varepsilon$. We now fix such an $\varepsilon$. By nondegeneracy, the three components of $e$ 
together with $\varepsilon$ form a basis of $\calV_x$ at every point $x$ of the boundary. This means that on the boundary we can expand $\Tilde e_n = \rho\epsilon + \xi^i e_i =\rho\varepsilon+\iota_\xi e$, where $\rho$ is a (nowhere vanishing) function and $\xi$ is interpreted as a vector field on $\Sigma$ ($\rho$ and $\xi$ are uniquely determined). Inserting this into \eqref{e:ELeden} evaluated on points on the boundary yields, after some algebra, 
\[
\iota_\xi(e\dd_\omega e)+
\rho\varepsilon\dd_\omega e + e((\dd_{\Tilde\omega}\Tilde e)_n-\iota_\xi\dd_\omega e) = 0.
\]
The first term vanishes upon using the constraint \eqref{e:edecon}, so we can ignore it. 
We then see that the additional constraint, callled the structural constraint in \cite{CCS20}, is 
\begin{equation}\label{e:strcon}
\varepsilon\dd_\omega e=e\sigma
\end{equation}
for some (1,1)\ndash form $\sigma$.\footnote{The evolution equation at boundary points can now be written in normal form as $(\dd_{\Tilde\omega}\Tilde e)_n=\iota_\xi\dd_\omega e-\rho\sigma$.} We then have, for whatever $\varepsilon$ we chose,
\[
\Hat C_\Sigma=\left\{(e,\omega)\in \Hat F^\de_{\Sigma} : eF_\omega+\frac\Lambda6e^3 = 0,\ e\dd_\omega e=0,\ 
\exists \sigma\ \varepsilon\dd_\omega e=e\sigma
\right\}.
\]

The crucial observation now is 
Theorem~17 in \cite{CCS20} which asserts that, for every metric nondegenerate 
$e$, there is a unique $\omega$ in each equivalence class satisfying the structural constraint \eqref{e:strcon}.\footnote{
In \cite{CCT20} it is shown how to modify the result in the light-like case.} 
This means that we can identify the quotient space $F^\de_{\Sigma}$ with the subspace
\[
F^\varepsilon_{\Sigma}\coloneqq\left\{
(e,\omega)\in \Hat F^\de_{\Sigma} : \exists \sigma\ \varepsilon\dd_\omega e=e\sigma\right\}
\]
of $\Hat F^\de_{\Sigma}$.
We therefore get the diagram
\begin{equation}\label{e:d:Fs}
\begin{tikzcd}
  F^\varepsilon_{\Sigma} \arrow[rd] \arrow[r, hook, "i"] \arrow[dr, hook, two heads, "\tau^\varepsilon"] & \Hat F^\de_{\Sigma}  \arrow[d, "\Hat p"]  \\
                                                                               & F^\de_{\Sigma}
\end{tikzcd} 
\end{equation}
where $i$ is the inclusion map and $\tau^\varepsilon$ is a diffeomorphism.
We then obtain the space of Cauchy data $C_\Sigma=\Hat p(\Hat C_\Sigma)$ as $\tau^\varepsilon(C_\Sigma^\varepsilon)$ with
\[
C_\Sigma^\varepsilon \coloneqq i(\Hat C_\Sigma) = \left\{(e,\omega)\in F^\varepsilon_{\Sigma} : eF_\omega+\frac\Lambda6e^3 = 0,\ e\dd_\omega e=0
\right\}.
\]
However, it is clear that $C_\Sigma^\varepsilon =i(\Hat C'_\Sigma)$, with $\Hat C'_\Sigma$ defined in \eqref{e:C'},
so 
$C_\Sigma=\tau^\varepsilon(C_\Sigma^\varepsilon)=\tau^\varepsilon(i(\Hat C'_\Sigma))=\Hat p(\Hat C'_\Sigma)$.
\end{proof}

We conclude this section 
with the following important
\begin{Thm}\label{t:PCCcoiso}
In the boundary metric nondegenerate PC theory,
 $C_\Sigma$ is coisotropic.
\end{Thm}
\begin{proof}
Thanks to Theorem~\ref{t:emnd}, we can work with $\Hat C'_\Sigma$ which we can rewrite as the common zero locus of the functions
\begin{align*}
P_c &= \int_\Sigma c\,e\dd_\omega e,\\
T_\mu &= \int_\Sigma \mu\, \left(eF_\omega+\frac\Lambda6e^3
\right),
\end{align*}
with $c$ a $(0,2)$\ndash form and $\mu$ a $(0,1)$\ndash form. We now compute their hamiltonian vector fields ($X_c$ and $Y_\mu$) with respect to the degenerate $2$\ndash form $\omega_M$ defined in \eqref{e:preomegaPC}: $\iota_{X_c}\omega_M=\delta P_c$, $\iota_{Y_\mu}\omega_M=\delta T_\mu$. The nontrivial fact---a consequence of Theorem~\ref{t:emnd}---is that they exist. They are not unique, but their ambiguity lies in the direction of the equivalence relation, along which the functions are invariant. As a consequence, their Poisson brackets are well-defined. We will only compute their restrictions to $C_\Sigma^\varepsilon$ and show that they vanish. This is enough to prove the theorem.

We present the hamiltonian vector fields and leave the rest of the computation to the reader:
\begin{align*}
X_c(e) &= c\cdot e, & X_c(\omega) &= \dd_\omega c + v,\\
Y_\mu(e) &\approx \dd_\omega\mu & eY_\mu(\omega) &= \mu\,\left(F_\omega +\frac\Lambda2e^2\right),
\end{align*}
where $(c\cdot e)_i^a = \eta_{rs}c^{ar}e_i^s$, $v$ satisfies $ev=0$, and $\approx$ means upon restriction to $C_\Sigma^\varepsilon$.
\end{proof}

\begin{Rem}
Looking into the proof, we see that $X_c$ generates the internal gauge transformations. On the other hand, one can see that $Y_\mu$ generates diffeomorphisms, both tangential and transversal to $\Sigma$. To see this, it is better to expand
$\mu = \lambda\varepsilon + \iota_\xi e$, with $\varepsilon$ as in the proof to Theorem~\ref{t:emnd}, where $\lambda$ is a function and $\xi$ a vector field on $\Sigma$. We then have $T_\mu=H_\lambda+P_\xi$ with
\begin{align*}
H_\lambda &= \int_\Sigma \lambda\,\varepsilon\, \left(eF_\omega+\frac\Lambda6e^3\right),\\
P_\xi &= \int_\Sigma \iota_\xi e\,eF_\omega.
\end{align*}
Computing the hamiltonian vector fields separately, one can see that that of $P_\xi$ generates the tangential vector fields and that of $H_\lambda$ the transversal ones. One can also recognize $H_\lambda$ and $P_\xi$ as the hamiltonian and momentum constraints. The full structure of the Poisson brackets among $P_C$, $H_\lambda$, and $P_\xi$ is described in \cite{CCS20}.
\end{Rem}

\begin{Rem}
In the computations in the proof of Theorem~\ref{t:PCCcoiso}, several integration by parts are involved. If one wants to extend this construction to the case when $\Sigma$ has a boundary, the boundary terms generate an interesting ``corner'' structure. This is better studied using the BFV formalism \cite{FV75,FK77}; see \cite{CC22}. An alternative approach would consist in extending to the context of gravity the results of \cite{RS22}.
\end{Rem}

\section{Conclusion}
In this note, we have described the construction of the reduced phase space for Lagrangian fields theories following the
geometrical method by Kijowski and W. M. Tulczyjew \cite{KT} and, in particular, we have applied it to the case of coframe gravity in four dimensions.


\end{document}